\documentclass[10pt, doublecolumn]{IEEEtran}
\pdfoutput=1

\usepackage{paralist}
\usepackage[usenames,dvipsnames]{xcolor}
\usepackage{amsthm,amsmath,amsfonts,ed}
\usepackage{xspace}
\usepackage{url}
\usepackage{multirow}
\usepackage{footnote}
\usepackage{subfigure}
\usepackage{amssymb,amsmath,mathtools,fixmath}
\usepackage{algorithmicx,algorithm}
\usepackage{algpseudocode} 
\usepackage{cases}
\usepackage{cite}
 
\newcommand{\x}{\mathbold{x}}
\newcommand{\w}{\mathbold{w}}

\newcommand{\y}{\mathbold{y}}
\newcommand{\bphi}{\mathbold{\phi}}
\newcommand{\bPhi}{\mathbold{\Phi}}
\newcommand{\bGamma}{\mathbold{\Gamma}}
\newcommand{\transp}{\mathsf{T}}
\newcommand{\Lap}{\mathbold{L}}
\newcommand{\F}{\mathbold{F}}
\newcommand{\M}{\mathbold{M}}
\newcommand{\A}{\mathbold{A}}
\newcommand{\B}{\mathbold{B}}
\newcommand{\C}{\mathbold{C}}
\renewcommand{\H}{\mathbold{H}}

\newcommand{\etal}{et~al.\xspace}
\newcommand{\eg}{e.g.,\/~}
\newcommand{\ie}{i.e.,\/~}

\newcommand{\wrt}{w.r.t.\xspace}                     
\newcommand{\ARMA}[1]{\textrm{ARMA}$_{#1}$\xspace}                    
\newcommand{\FIR}[1]{\textrm{FIR}$_{#1}$\xspace} 
\newcommand{\hsk}{\hskip-0.05cm}                   
 
\renewcommand{\paragraph}[1]{\vspace{2mm}\noindent \textbf{#1}}

\DeclarePairedDelimiter{\abs}{\lvert}{\rvert}

\newcommand{\lmin}{\lambda_{\textit{min}}}
\newcommand{\mumin}{\mu_{\textit{min}}}
\newcommand{\lmax}{\lambda_{\textit{max}}}
\newcommand{\mumax}{\mu_{\textit{max}}}

\newtheoremstyle{slplain}
  {0.7\baselineskip\@plus.2\baselineskip\@minus.2\baselineskip} 
  {0.0\baselineskip\@plus.2\baselineskip\@minus.2\baselineskip}
  {\slshape}												  
  {}														  
  {\itshape}												  
  {.}													      
  { }														  
  {}													      
\theoremstyle{slplain}

\newtheorem{proposition}{Proposition}
\newtheorem{corollary}{Corollary}
\newtheorem{remark}{Remark}

\makesavenoteenv{tabular}

\begin{document}

\title{Distributed Autoregressive Moving Average\\ Graph Filters}

\author{Andreas Loukas*, Andrea Simonetto, and Geert Leus
\thanks{The authors are with the Faculty of EEMCS, Delft University of Technology, 2826 CD Delft, The Netherlands. e-mails: \{a.loukas, a.simonetto, g.j.t.leus\}@tudelft.nl. *Corresponding author: A. Loukas.}}

\maketitle

\begin{abstract}
We introduce the concept of autoregressive moving average (ARMA) filters on a graph and show how they can be implemented in a distributed fashion. Our graph filter design philosophy is independent of the particular graph, meaning that the filter coefficients are derived irrespective of the graph. 
In contrast to finite-impulse response (FIR) graph filters, ARMA graph filters are robust against changes in the signal and/or graph. In addition, when time-varying signals are considered, we prove that the proposed graph filters behave as ARMA filters in the graph domain and, depending on the implementation, as first or higher ARMA filters in the time domain. 
\end{abstract}


\section{Introduction}

The emerging field of signal processing on graphs~\cite{Sandryhaila2014a, Sandryhaila2013, Sandryhaila2014, Shuman2013} focuses on the extension of classical discrete signal processing techniques to the graph setting. Arguably, the greatest breakthrough of the field has been the extension of the Fourier transform from time signals and images to graph signals, i.e., signals defined on the nodes of irregular graphs. By providing a graph-specific definition of frequency, the graph Fourier transform (GFT) enables us to design filters for graphs:
analogously to classical filters, graph filters process a graph signal by amplifying or attenuating its components at specific graph frequencies.
Graph filters have been used for a number of signal processing tasks, such as denoising~\cite{Zhang2008,Chen2015}, centrality computation~\cite{Page1999}, graph partitioning~\cite{Chung2007}, event-boundary detection~\cite{Loukas2014}, and graph scale-space analysis~\cite{Loukas2015}.

Distributed implementations of filters on graphs only emerged recently as a way of increasing the scalability of computation~\cite{Shuman2011,Sandryhaila2014,Safavi2014}. Nevertheless, being inspired by finite impulse response (FIR) graph filters, these methods are sensitive to graph changes. To solve the graph robustness issue, distributed infinite impulse response (IIR) graph filters have been proposed by Shi~\etal\cite{Shi2015}. 
Compared to FIR graph filters, IIR filters have the potential to achieve better interpolation or extrapolation properties around the known graph frequencies. Moreover, by being designed for a continuous range of frequencies, they can be applied to any graph (even when the actual graph spectrum is unknown).

In a different context, we introduced graph-independent IIR filter design, or what we will label here as {\it universal} IIR filter design (in fact, prior to \cite{Shi2015}) using a potential kernel approach~\cite{Loukas2013,Loukas2014}. In this letter, we will build upon our prior work to develop more general autoregressive moving average (ARMA) graph filters of any order, using parallel or periodic concatenations of the potential kernel. This leads to a {more intuitive} distributed design than the one proposed by Shi~\etal, which is based on gradient-descent type of iterations. Moreover, we show that the proposed ARMA graph filters are suitable to handle \emph{time-varying signals}, an important issue that was not considered previously. Specifically, our design extends {naturally} to time-varying signals leading to 2-dimensional ARMA filters: an ARMA filter in the graph domain of arbitrary order and a first order AR (for the periodic implementation) or a higher order ARMA (for the parallel implementation) filter in the time domain; which opens the way to a deeper understanding of graph signal processing, in general. We conclude the letter by displaying preliminary results suggesting that our ARMA filters not only work for continuously time-varying signals but are also robust to continuously \emph{time-varying graphs}.

\section{Graph Filters}

Consider a graph $G = (V,E)$ of $N$ nodes and let $\x$ be a signal defined on the graph, whose $i$-th component represents the value of the signal at the $i$-th node\footnote{We denote the $i$-th component of a vector $\x$ as $x_i$ starting at index $1$. Node $i$ of a graph is denoted as $u_i$.}.

\vskip-1mm\paragraph{Graph Fourier Transform (GFT).} The GFT transforms a graph signal into the graph frequency domain: the forward and inverse GFTs of $\x$ are $\hat{x}_n = \langle \x , \bphi_n \rangle $ and $x_n =  \langle \hat{\x}, \bphi_n \rangle$, where $\langle \,, \rangle$ denotes the inner product. Vectors $\{\bphi_n\}_{n = 1}^N$ form an orthonormal basis and are commonly chosen as the eigenvectors of a graph Laplacian $\Lap$, such as the discrete Laplacian $\Lap_{\textrm{d}}$ or Chung's normalized Laplacian $\Lap_{\textrm{n}}$.  
For an extensive review of the properties of the GFT, we refer to~\cite{Shuman2013,Sandryhaila2014}.

\emph{To avoid any restrictions on the generality of our approach}, in the following we present our results for a \emph{general basis matrix} $\Lap$. We only require that $\Lap$ is \emph{symmetric} and \emph{1-local}: for all $i \neq j$, $\Lap_{ij} = 0$ whenever $u_i$ and $u_j$ are not neighbors and $L_{ij} = L_{ji}$ otherwise. 

\vskip-1mm\paragraph{Graph filters.} A \emph{graph filter} $\F$ is a linear operator that acts upon a graph signal $\x$ by amplifying or attenuating its graph Fourier coefficients as 
\begin{equation}
	\F \x = \sum\limits_{n = 1}^N h(\lambda_n) \, \hat{x}_n \bphi_n.
\end{equation} 
Let $\lmin$ and $\lmax$ be the minimum and maximum eigenvalues of $\Lap$ over \emph{all} possible graphs. The graph frequency response $h : [\lmin, \, \lmax] \rightarrow \mathbb{C}$ controls how much $\F$ amplifies the signal component of each graph frequency
\begin{align}
	h(\lambda_n) = \langle \F\x, \bphi_n \rangle / \langle \x, \bphi_n \rangle.
\end{align} 

\vskip-1mm\paragraph{Distributed graph filters.} We are interested in how we can filter a signal with a graph filter $\F$ in a \emph{distributed} way, having a user-provided frequency response $h^\ast(\lambda)$. Note that this prescribed $h^\ast(\lambda)$ is a continuous function in the graph frequency $\lambda$ and describes the desired response for {\em any} graph. The corresponding filter coefficients are thus independent of the graph and universally applicable. 

\vskip-1mm\paragraph{\FIR{K} filters.} It is well known that we can \emph{approximate} $\F$ in a distributed way by using a $K$-th order polynomial of $\Lap$. Define \FIR{K} as the $K$-th order approximation given by 
\begin{equation*}
	\F_K = h_0 {\bf I} + \sum_{k=1}^K h_k \Lap^{k}, 
\end{equation*}
where the coefficients $h_i$ are found by minimizing the least-squares objective $\int_{\lambda} | \sum_{k=0}^K h_k \lambda^{k} - h^\ast(\lambda)|^2 \mathrm{d} \lambda$. 
Observe that, in contrast to traditional graph filters, the order of the considered {\it universal} graph filters is not necessarily limited to $N$. By increasing $K$, we can approximate any filter with square integrable frequency response arbitrarily well.

The computation of \FIR{K} is easily performed distributedly. Since $\Lap^K \x = \Lap \left(\Lap^{K-1} \x \right) $, each node $u_i$ can compute the $K$th-term from the values of the $(K-1)$th-term in its neighborhood. The algorithm terminates after $K$ iterations, and, in total, each node exchanges $\Theta(K \deg{u_i})$ bits and stores $\Theta(\deg{u_i} + K)$ bits in its memory. However, \FIR{K} filters exhibit poor performance when the signal or/and graph are time-varying and when there exists asynchronicity among the nodes\footnote{This because, first the distributed averaging is paused after $K$ iterations, and thus the filter output is \emph{not a steady state}; second the input signal is only considered during the first iteration. To track time-varying signals, the computation should be restarted at each time step, increasing the communication and space complexities to $\Theta(K^2\deg{u_i})$ bits and $\Theta(K \deg{u_i} + K^2)$ bits.}. In order to overcome these issues and provide a more solid foundation for graph signal processing, we study \ARMA{} graph filters.

\section{ARMA Graph Filters}

\subsection{Distributed computation}
\label{subsec:filter_computation}

We start by presenting a simple recursion that converges to a filter with a 1st order rational frequency response. We then propose two generalizations with $K$-th order responses\footnote{Note that similar structures were independently developed in \cite{Shi2015}, although based on a different design methodology.}. Using the first, which entails running $K$ 1st order filters in parallel, a node $u_i$ attains fast convergence at the price of exchanging and storing $\Theta(K \deg{u_i})$ bits per iteration\footnote{Any values stored are overwritten during the next iteration.}. By using periodic coefficients, the second algorithm reduces the number of bits exchanged and stored to $\Theta(\deg{u_i})$, at almost equivalent (or even \emph{faster}) convergence time. 

\vskip-1mm\paragraph{\ARMA{1} filters.} We will obtain our first ARMA graph filter as an extension of the potential kernel~\cite{Loukas2013}. Consider the following 1st order recursion:
\begin{align}
	\y_{t+1}   &= \psi \M \y_{t} + \varphi \x \quad \text{and} \quad \y_{0} \text{ arbitrary}, 		
	\label{eq:potentials}
\end{align}
where the coefficients $\varphi, \psi$ are (for now) arbitrary complex numbers, and $\M$ is the translation of $\Lap$ with the minimal spectral radius: $\M = \frac{\lmax - \lmin}{2}{\bf I} - \Lap$.
From Sylvester's matrix theorem, matrices $\M$ and $\Lap$ have the same eigenvectors and the eigenvalues $\mu_n$ of $\M$ differ by a translation to those of $\Lap$: $\mu_n = (\lmax-\lmin)/2 - \lambda_n$.

\begin{proposition}
	The frequency response of \ARMA{1} is $g(\mu) = \frac{r}{\mu - p}, \quad s.t. \quad \abs{p} > \frac{\lmax-\lmin}{2}$, with the residue $r$ and the pole $p$ given by $r = -\varphi / \psi $ and $p = 1/ \psi$, respectively. Recursion~\eqref{eq:potentials} converges to it linearly, irrespective of the initial condition $\y_0$ and matrix $\Lap$.  
	\label{prop:arma1}
\end{proposition}
\begin{proof}
The proof follows from Theorem~1 in~\cite{Loukas2013}, in which we replace $P$ with $\M$ and $1-\varphi$ with $\psi$.
\end{proof}

Recursion~\eqref{eq:potentials} leads to a very efficient distributed implementation: at each iteration $t$, each node $u_i$ updates its value $y_{t,i}$ based on its local signal $x_i$ and a weighted combination of the values $y_{t-1,j}$ of its neighbors $u_j$. Since each node must exchange its value with each of its neighbors, the message/space complexity at each iteration is $\Theta(\deg{u_i})$ bits.

\paragraph{Parallel \ARMA{K} filters.} We can attain a larger variety of responses by simply adding the output of multiple 1st order filters. 
Denote with the superscript $k$ the terms that correspond to the $k$-th \ARMA{1} filter ($k=1,2,\dots,K)$. 

\begin{corollary}
The frequency response of a parallel \ARMA{K} is 
\begin{equation*}
	g(\mu) = \sum_{k = 1}^K \frac{r^{(k)}}{\mu - p^{(k)}} \quad \text{s.t.} \quad |{p^{(k)}}| > \frac{\lmax-\lmin}{2},
\end{equation*}	
with $r^{(k)} = -\varphi^{(k)} / \psi^{(k)} $ and $p^{(k)} = 1/ \psi^{(k)}$, respectively. Recursion~\eqref{eq:potentials} converges to it linearly, irrespective of the initial condition $\y_0$ and matrix $\Lap$.
\label{corollary:arma_parallel}
\end{corollary}
\begin{proof}\emph{(Sketch)}
From Proposition~\ref{prop:arma1}, at steady state, we have
\begin{equation*}
	\y = \sum_{k = 1}^K \y^{(k)} = \sum_{k = 1}^K \sum_{n = 1}^N \left( \frac{r^{(k)}}{\mu_n - p^{(k)}} \right)  \hat{x}_n \bphi_n,
\end{equation*}
and switching the sum operators the claim follows. \end{proof}

The frequency response of a parallel \ARMA{K} is therefore a rational function with numerator and denominator polynomials of orders $K-1$ and $K$, respectively\footnote{By choosing the coefficients properly, we can generalize the rational function to have any degree smaller than $K$ in the numerator. By adding an extra input, we can also obtain order $K$ in the numerator.}. At each iteration, node $u_i$ exchanges and stores $\Theta(K \deg{u_i})$ bits.

\paragraph{Periodic \ARMA{K} filters.} We can decrease the memory requirements of the parallel implementation by letting the filter coefficients vary in time.
Consider the output of the time-varying recursion
\begin{align}
	\y_{t+1}   &= (\theta_t {\bf I} + \psi_t \M) \y_{t} + \varphi_t \x \quad \text{and} \quad \y_{0} \quad \text{arbitrary},
\label{eq:armaperiodic}
\end{align}
every $K$ iterations, where coefficients $\theta_t, \psi_t, \varphi_t$ are periodic with period $K$:	$\theta_t = \theta_{t-iK}, \psi_t = \psi_{t-iK}, \varphi_t = \varphi_{t-iK}$, with $i$ an integer in $[0, t/K]$ and $\theta_t = 1 - \mathrm{III}_{K}(t)$ being the negated Shah function. 

\begin{proposition}
The frequency response of a periodic \ARMA{K} filter is 
\begin{align*}
	g(\mu) &= \frac{ \sum_{\tau = 0}^{K-1} \prod_{\sigma = K-\tau}^{K-1} \left(\theta_{\sigma} + \psi_{\sigma}\mu\right) \varphi_{K-\tau-1}}{1 - \left(\prod_{\tau = 0}^{K-1} \theta_{\tau} + \psi_{\tau}\mu\right)},
\end{align*}
s.t. the stability constraint $| \prod_{\tau = 0}^{K-1} \theta_{\tau} + \psi_{\tau} \frac{\lmax - \lmin}{2} | <  1$. Recursion~\eqref{eq:armaperiodic} converges to it linearly, irrespective of the initial condition $\y_0$ and matrix $\Lap$.
\label{prop:arma_periodic}
\end{proposition}

\begin{proof}
Define matrices $\bGamma_t = \theta_t {\bf I} + \psi_t \M$ and $\bPhi_{t_1,t_2} = \bGamma_{t_1} \bGamma_{t_1-1} \cdots \bGamma_{t_2}$ if $t_1 \geq t_2$ and $\bPhi_{t_1,t_2} = {\bf I}$ otherwise. 
The output at the end of each period can be re-written as a time-invariant system 
\begin{align}
	\y_{(i+1)K} = \A \y_{iK} + \B \x,
	\label{eq:recursion_coarse}
\end{align}
with $\A = \bPhi_{K-1,0}$, 
$ \B = \sum_{\tau = 0}^{K-1} \bPhi_{K-1,K-\tau} \varphi_{K-\tau-1}$.
Assuming that $\A$ is non-singular, both $\A$ and $\B$ have the same eigenvectors $\bphi_n$ as $\M$ (and $\Lap$). As such, when $\abs*{\lmax(\A)} < 1$, the steady state of~\eqref{eq:recursion_coarse} is 
\begin{equation*}
 \y = (I-\A)^{-1} \B \x = \sum_{n=1}^N \frac{\lambda_n(\B)}{1 - \lambda_n(\A)} \hat{x}_n \bphi_n.
\end{equation*}
To derive the exact response, notice that 
\begin{align*}
	\lambda_n(\bPhi_{t_1,t_2}) = \prod_{\tau = t_1}^{t_2} \lambda_n(\bGamma_t) =  \prod_{\tau = t_1}^{t_2} \left(\theta_{\tau} + \psi_{\tau}\mu_n \right),
\end{align*}
which, by the definition of $\A$ and $\B$, yields the desired frequency response. 
The linear convergence rate follows from the linear convergence of~\eqref{eq:recursion_coarse} to $\y$ with rate $\gamma = |\lmax(\A)|$.
\end{proof}

\vspace{2mm} By some algebraic manipulation, we can see that the frequency responses of periodic and parallel \ARMA{K} filters are equivalent at steady state. 
In the periodic version, each node $u_i$ stores $\Theta(\deg(u_i))$ bits, as compared to $\Theta(K \deg{u}_i)$ bits in the parallel one. The low-memory requirements of the periodic \ARMA{K} render it suitable for resource constrained devices. 

\begin{remark}
Since the designed \ARMA{K} filters are attained for any initial condition and matrix $\Lap$, the filters are also robust to slow time-variations in the signal and graph. We will generalize this result to arbitrary time-varying signals in Section~\ref{sec:tv}.  
\end{remark}

\subsection{Filter design}
\label{subsec:filter_design}

Given a graph frequency response $g^\ast:$ $[\mumin,\, \mumax ]$ $\rightarrow \mathbb{C}$ and a filter order $K$, our objective is to find the complex polynomials $p_b(\mu)$ and $p_a(\mu)$ of order $K-1$ and $K$, respectively, that minimize
\begin{align*}
	\int_{\mu} \hsk\Big|{\frac{p_b(\mu)}{p_a(\mu)} \hsk-\hsk g^\ast(\mu)}\Big|^2\hsk\hsk\mathrm{d}\mu \hsk=\hsk \int_{{\mu}}\hsk\Big|\frac{ \sum_{k=0}^{K-1} b_k \mu^k }{1 \hsk+\hsk \sum_{k=1}^{K} a_k \mu^k} \hsk-\hsk g^\ast(\mu)\Big|^2\hsk\hsk\mathrm{d}\mu,
\end{align*}
while ensuring that the chosen coefficients result in a stable system (see constraints in Corollary~\ref{corollary:arma_parallel} and Proposition~\ref{prop:arma_periodic}).

\begin{remark} Whereas $g^\ast$ is a function of $\mu$, the desired frequency response $h^\ast: [\lmin, \, \lmax] \rightarrow \mathbb{C}$ is often a function of $\lambda$. We attain $g^\ast(\mu)$ by simply mapping the user-provided response to the domain of $\mu$: $g^\ast(\mu) = h^\ast((\lmax-\lmin)/2-\lambda)$. 
\end{remark}

\begin{remark} Even if we constrain ourselves to pass-band filters and we consider only the set of $\Lap$ for which $(\lmax\hsk-\hsk\lmin)/2 \hsk=\hsk1$, it is impossible to design our coefficients based on classical design methods developed for IIR  filters (\eg Butterworth, Chebyshev). The stability constraint of \ARMA{K} is different from classical filter design, where the poles of the transfer function must lie within (not outside) the unit circle.
\end{remark}

\begin{figure}
\includegraphics[width=1\columnwidth]{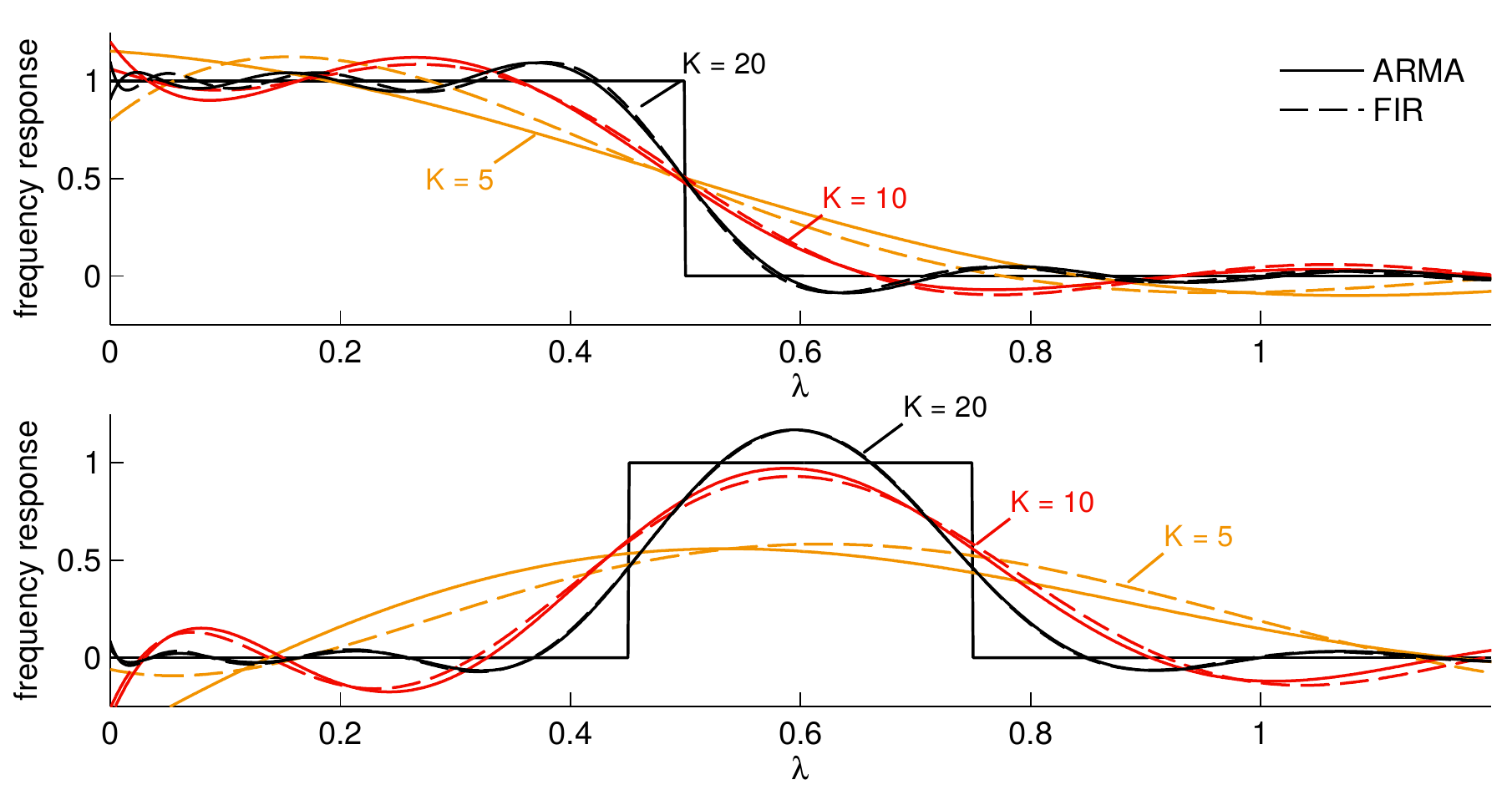}
\vskip-1mm
\caption{The frequency response of \ARMA{K} filters designed by Shank's method and the FIR responses of corresponding order. Here, $h^\ast$ is a step function (top) and a window function (bottom). }
\label{fig:responses}
\vskip-1mm
\end{figure}

\begin{figure}[t]
\centering
\includegraphics[width=1\columnwidth]{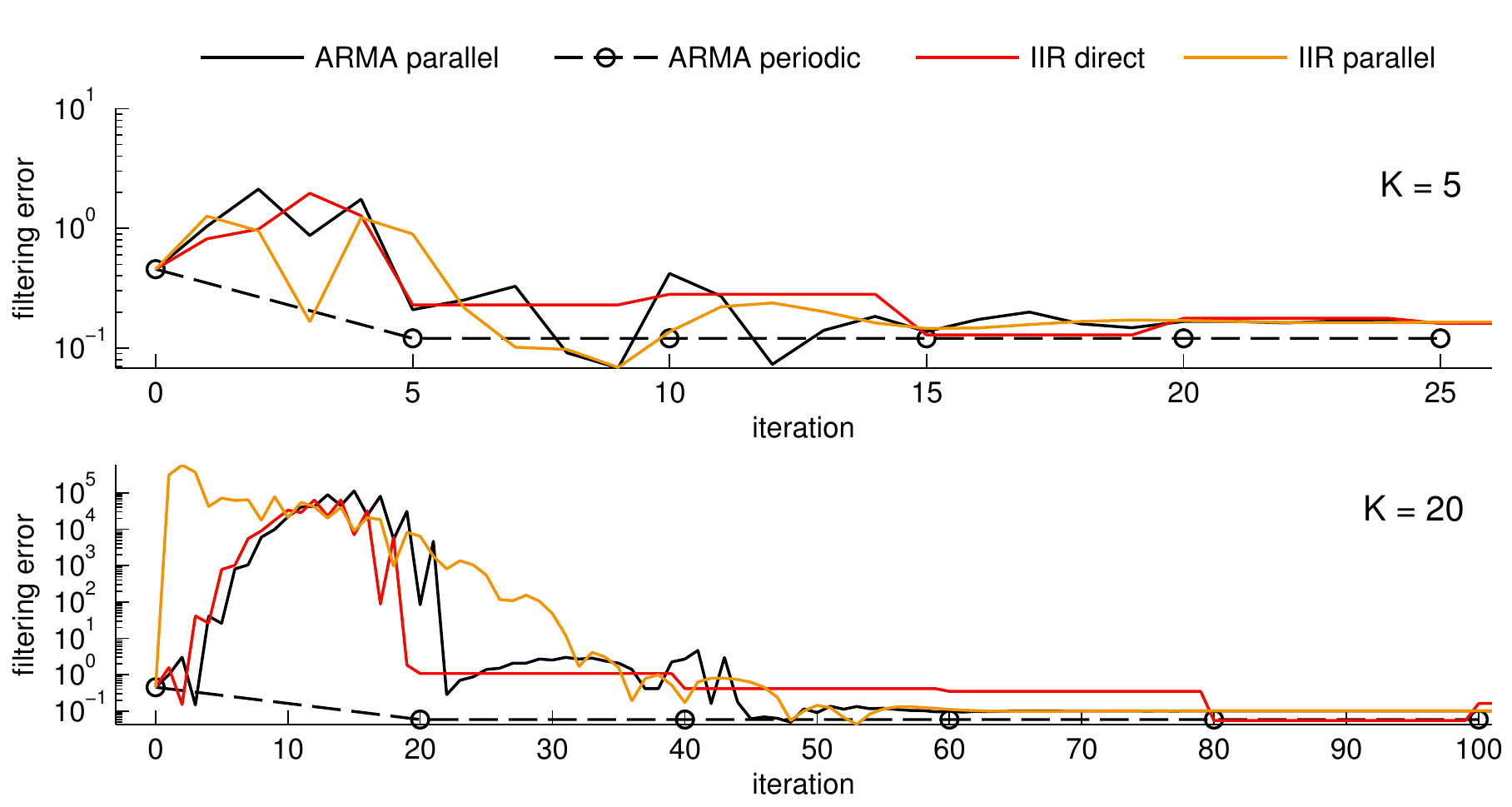}
\vskip-1mm
\caption{Convergence comparison of \ARMA{} filters w.r.t. the IIR filters of~\cite{Shi2015}. The filtering error is $\|\y_t-\y^*\|_2/\|\y^*\|_2$, where $\y^*$ is the desired output.}
\label{fig:comparison}
\vskip-1mm
\end{figure}

\vskip-1mm\paragraph{Design method.} Similar to Shank's method~\cite{Shanks1967}, we approximate the filter coefficients in two steps:

1) We determine $\{a_k\}_{k=1}^K$, by finding a $\hat{K} > K$ order polynomial approximation $\hat{g}(\mu) = \sum_{k=0}^{\hat{K}} g_k \mu^k$ of $g^\ast(\mu)$ using polynomial regression, and solving the coefficient-wise system of equations $p_a(\mu) \hat{g}(\mu) = p_b(\mu)$. 

2) We determine $\{b_k\}_{k=1}^{K-1}$ by solving the constrained least-squares problem of minimizing $\int_{{\mu}}|p_b(\mu)/p_a(\mu) - g^\ast(\mu)|^2 \mathrm{d}\mu$, \wrt $p_b(\mu)$ and s.t. the stability constraints.

\vspace{1mm} Figure~\ref{fig:responses} illustrates in solid lines the frequency responses of three \ARMA{K} filters ($K=5,10,20$), designed to approximate a step function (top) and a window function (bottom). In the first step of our design, we computed the FIR filter $\hat{g}$ as a Chebyshev approximation of $g^\ast$ of order $\hat{K} = K+1$. ARMA responses closely approximate the optimal FIR responses for the corresponding orders (dashed lines).

Figure~\ref{fig:comparison} compares the convergence of our recursions w.r.t. the IIR design of~\cite{Shi2015} in the same low-pass setting of Figure~\ref{fig:responses} (top), running in a network of $n=100$ nodes\footnote{We do not consider the cascade from of~\cite{Shi2015} since every module in the cascade requires many iterations, leading to a slower implementation.}. We see how our periodic implementation (only valid at the end of each period) obtains faster convergence. The error of other filters increases significantly at the beginning for $K=20$, due to the filter coefficients, which are very large.

\section{Time variations}
\label{sec:tv}

We now focus on \ARMA{K} graph filters and study their behavior when the signal is changing in time, thereby showing how our design extends naturally to the analysis of time-varying signals. We start by \ARMA{1} filters: indicate with $\x_t$ the graph signal at time $t$. We can re-write the \ARMA{1} recursion as 
\begin{equation}\label{eq:lti2}
	\y_{t+1} = \psi \M \y_{t} + \varphi \x_t. 
\end{equation}
The graph signal $\x_t$ can still be decomposed into its graph Fourier coefficients, only now they will be time-varying, \ie we will have $\hat{x}_{n,t}$. Under the stability condition $\|\psi \M\| < 1$, for each of these coefficients we can write its respective graph frequency \emph{and} standard frequency transfer function as
\begin{equation}
H(z,\mu) = \frac{\varphi}{z - \psi \mu}.
\end{equation}
The transfer functions ${H}(z,\mu)$ characterize completely the behavior of \ARMA{1} graph filters for an arbitrary yet time-invariant graph: when $z \to 1$, we obtain back the constant $\x$ result of Proposition~\ref{prop:arma1}, while for all the other $z$ we obtain the standard frequency response as well as the graph frequency one. As one can see, 1st order filters are universal \ARMA{1} in the graph domain (they do not depend on the particular choice of $\Lap$) as well as 1st order AR filters in the time domain. This result generalizes to parallel and periodic \ARMA{K} filters. 

\vskip-0mm\paragraph{Parallel \ARMA{K}.} Similarly to Corollary~\ref{corollary:arma_parallel}, we have: 

\begin{proposition}
Under the same stability conditions of Corollary~\ref{corollary:arma_parallel}, the transfer function $H(z,\mu)$ from the input ${\x}_{t}$ to the output $\y_t$ of a parallel \ARMA{K} implementation is 
\begin{equation*}
	H(z, \mu) = \sum_{k = 1}^K \frac{\varphi^{(k)}}{z - \psi^{(k)}\mu}.
\end{equation*}	
\label{prop.tvpar}
\end{proposition}
\begin{proof}
The recursion~\eqref{eq:potentials} for the parallel implementation reads 
\begin{equation}
\y_{t+1}^{(k)} = \psi^{(k)}\M \y_t^{(k)} + \varphi^{(k)} \x_t, \quad k=1,\dots,K
\end{equation}
while the output is $\y_{t} = \sum_{k=1}^K \y_t^{(k)}$. This can be written in a compact form as 
\begin{equation}
\w_{t+1} = \A \w_t + \B \x_t, \quad \y_{t} = \C \w_t,
\end{equation}
where $\w_t$ is the stacked version of all the $\y^{(k)}_t$, while 
$$
\A = \mathrm{blkdiag}[\psi^{(1)}\M, \dots, \psi^{(K)}\M], \,\B = [\varphi^{(1)} {\bf I}, \dots, \varphi^{(K)} {\bf I}]^\transp, 
$$
and $\C = \bf{1}^\transp \otimes {\bf I}$. Under the same stability conditions of Corollary~\ref{corollary:arma_parallel}, the transfer matrix between $\x_t$ and $\y_t$ is 
$$
\H(z) = \C ({z \bf{I} - \A})^{-1} \B = \sum_{k = 1}^K {\varphi^{(k)}}({z {\bf I} - \psi^{(k)}\M })^{-1},
$$
where we have used the block diagonal structure of $\A$. By applying the Graph Fourier transform, the claim follows.
\end{proof}

Proposition~\ref{prop.tvpar} characterizes the parallel implementation completely: our filters are universal \ARMA{K} in the graph domain as well as in the time domain.

\vskip-1mm\paragraph{Periodic \ARMA{K}.} Time-varying signals in the periodic implementation will be analyzed assuming that we keep the input $\x_t$ \emph{fixed} during the whole period $K$. 

\begin{figure}[t]
\centering
\includegraphics[width=1\columnwidth]{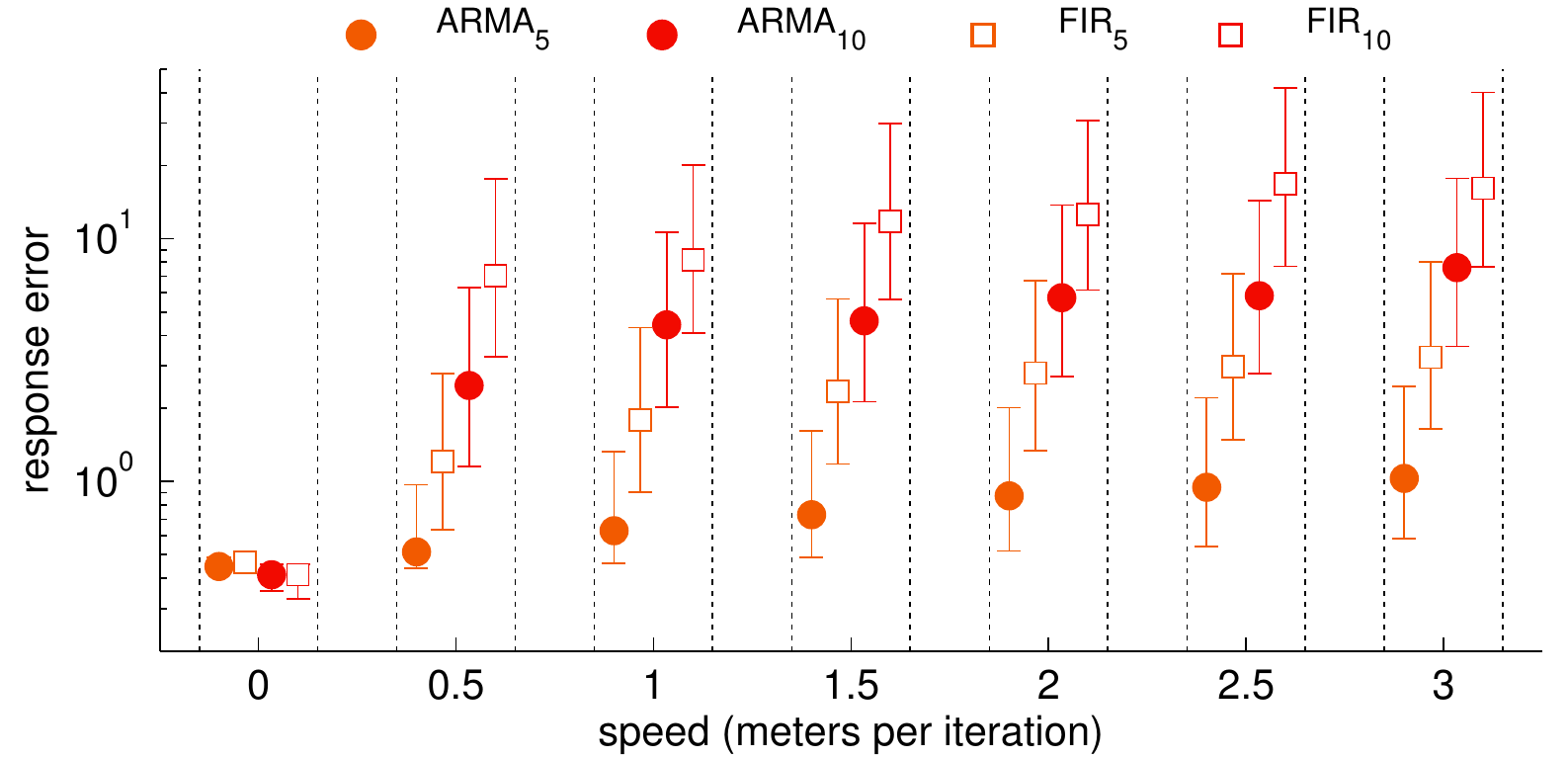}
\vskip-1mm
\caption{The effect of node mobility inducing a time-varying signal and graph. Each error bar depicts the standard deviation of the filtering error over ten runs. The response error is $\|g(\mu) - g^*(\mu)\|_2/\|g^*(\mu)\|_2$. A small horizontal offset was included to improve visibility.}
\label{fig:mobility}
\vskip-1mm
\end{figure}

\begin{proposition}
Let $\x_{iK}$ be a sampled version of the input signal $\x_t$, sampled at the beginning of each period. Under the same stability conditions of Proposition~\ref{prop:arma_periodic}, the transfer function for periodic \ARMA{K} filters from $\x_{iK}$ to $\y_{iK}$ is
\begin{equation}
H_K(z,\mu) = \frac{ \sum_{\tau = 0}^{K-1} \prod_{\sigma = K-\tau}^{K-1} \left(\theta_{\sigma} + \psi_{\sigma}\mu\right) \varphi_{K-\tau-1}}{z - \left(\prod_{\tau = 0}^{K-1} \theta_{\tau} + \psi_{\tau}\mu\right)}.
\end{equation}
\end{proposition}

\begin{proof}\emph{(Sketch)} One writes the recursion~\eqref{eq:recursion_coarse} substituting $\x$ with $\x_{Kt}$, and proceeds as in the proof of Proposition~\ref{prop:arma_periodic}. 
\end{proof}

As in the parallel case, this proposition describes completely the behavior of the periodic implementation. In particular, our filters are \ARMA{K} filters in the graph domain whereas 1st order AR filters in the time domain. 

\vskip2mm

The design of ${H}(z,\mu)$ and $H_K(z,\mu)$ to accommodate both \ARMA{K} requirements and bandwidth for time-varying signals is left for future research.

\vskip-0mm\paragraph{Time-varying graphs.} We conclude the letter with a preliminary result showcasing the robustness of our filter design to continuously time-varying signals \emph{and} graphs. Under the same setting of Figure~\ref{fig:responses}, we consider $\x_t$ to be the node degree, while moving the nodes by a random waypoint model~\cite{BonnMotion2010} for a duration of $600$ seconds. In this way, by defining the graph as a disk graph, the graph and the signal are changing. In Figure~\ref{fig:mobility}, we depict the response error after $100$ iterations (\ie at convergence), in different \emph{mobility} settings: the speed is defined in meters per iteration and the nodes live in a box of $1000\times1000$ meters with a communication range of $180$ meters. As we observe, our designs can tolerate better time-variations. Future research will focus on characterizing and exploiting this property from the design perspective. 

\bibliographystyle{IEEEtran}
\bibliography{bibliography}

\end{document}